\documentclass[prx,twocolumn,superscriptaddress,longbibliography,notitlepage]{revtex4-2}

\usepackage[english]{babel}
\usepackage[colorlinks=false,linkcolor=blue,citecolor=red,plainpages=false,pdfpagelabels]{hyperref}

\usepackage{multirow}
\usepackage{textcomp}
\usepackage{url}
\usepackage{tcolorbox}
\usepackage{physics}

\usepackage{lmodern}
\usepackage{amsfonts}
\usepackage{amsmath}
\usepackage{amssymb}
\usepackage{latexsym}
\usepackage{mathtools}

\usepackage[margin=1in]{geometry}
\usepackage{hyperref}
\hypersetup{pdfpagemode=UseNone}
\usepackage[capitalise]{cleveref}

\usepackage{amsthm}
\newtheorem{theorem}{Theorem}
\newtheorem{lemma}{Lemma}

\newtheorem{corollary}{Corollary}
\theoremstyle{definition}
\newtheorem{definition}{Definition}

\def\O{\mathcal{O}}

\def\T{\mathcal{T}}

\newcommand\wt[1]{\overline{#1}}
\newcommand{\loc}{{\text{loc}}}
\newcommand{\const}{{\text{const.}}}

\makeatletter
\let\@fnsymbol\@arabic
\makeatother

\renewcommand{\O}[1]{\mathcal{O}\!\left(#1\right)}
\newcommand{\Otilde}[1]{\tilde{\mathcal{O}}\!\left(#1\right)}

\begin{document}

\title{Hamiltonian Simulation in the Interaction Picture using the Magnus Expansion}

\author{Kunal Sharma}
\email{kunals@ibm.com}
\affiliation{IBM~Quantum,~IBM~T.J.~Watson~Research~Center,~Yorktown~Heights,~NY~10598,~USA}
\author{Minh C. Tran}
\email{minhtran@ibm.com}
\affiliation{IBM~Quantum,~IBM~T.J.~Watson~Research~Center,~Yorktown~Heights,~NY~10598,~USA}

\newpage

\begin{abstract}
We propose an algorithm for simulating the dynamics of a geometrically local Hamiltonian~$A$ under a small geometrically local perturbation~$\alpha B$.
In certain regimes, the algorithm achieves the optimal scaling and outperforms the state-of-the-art algorithms. 
By moving into the interaction frame of $A$ and classically computing the Magnus expansion of the interaction-picture Hamiltonian, our algorithm bypasses the need for ancillary qubits.
In analyzing its performance, we develop a framework to capture the quasi-locality of the Magnus operators, leading to a tightened bound for the error of the Magnus truncation.
The Lieb-Robinson bound also guarantees the efficiency of computing the Magnus operators and of their subsequent decomposition into elementary quantum gates.
These features make our algorithm appealing for near-term and early-fault-tolerant simulations.
\end{abstract}

\maketitle

\section{Introduction} \label{sec:introduction}

Hamiltonian simulation is one of the most exciting applications of quantum computers. 
Their universality guarantees that an arbitrary unitary evolution can be decomposed into a series of quantum gates.
Leading quantum algorithms for such decomposition include product formulas~\cite{Suzuki1991,Childs2019NearlyOptimal,childs_theory_2021}, linear combinations of unitaries~\cite{childs_hamiltonian_nodate,berrySimulatingHamiltonianDynamics2014new}, quantum signal processing~\cite{lowOptimalHamiltonianSimulation2017new,lowHamiltonianSimulationQubitization2019new}, and stochastic quantum simulation algorithms~\cite{childsFasterQuantumSimulation2019new,campbellRandomCompilerFast2019new,ouyangCompilationStochasticHamiltonian2020new}.
These algorithms each have their own advantages depending on the structure of the Hamiltonian.
High-order product formulas, for example, achieve nearly optimal scaling in gate complexity when applied to geometrically local Hamiltonians~\cite{childs_theory_2021,haah2021quantum}, while qDRIFT \cite{campbellRandomCompilerFast2019new} excels when there are multiple energy scales in the Hamiltonian.

In this paper, we consider a family of Hamiltonians $H = A + \alpha B$, where $\alpha B$ represents a small perturbation to the unperturbed Hamiltonian $A$.
Treatment of such Hamiltonians in the perturbation theory has led to important results, including Fermi's golden rule and quantum chemistry beyond the Born-Oppenheimer approximation.
Promising a systematic pathway to explore Hamiltonian dynamics under perturbation, various quantum simulation algorithms have been introduced.
A popular approach is to move into the interaction frame of $A$ and use time-dependent quantum simulation algorithms to simulate the resulting interaction-picture Hamiltonians~\cite{low_hamiltonian_2019,an_time-dependent_2022,rajputHybridizedMethodsQuantum2022new,chenQuantumAlgorithmTimedependent2021new}.
This approach often requires ancillary qubits and controlled-evolutions, making it challenging to implement ons near-term quantum computers.
Stochastic algorithms can bypass these resource requirements, but their asymptotic complexity is only comparable to that of low-order deterministic algorithms.

We propose an algorithm to simulate the evolution under a geometrically local Hamiltonian $H = A + \alpha B$.
The algorithm approximates the evolution in the interaction frame of $A$ by a truncated Magnus expansion.
Each Magnus operator is further truncated based on the Lieb-Robinson bound, allowing an efficient decomposition into the Pauli basis.
The evolution generated by the truncated Magnus operators is then implemented using product formulas.
Compared to existing results, our algorithm displays several favorable features, including using no ancillas, being independent of how fast $H$ changes with time, and achieving the optimal gate complexity in certain regimes.

\section{Setup}\label{sec:setup}

Given a system of $n$ qubits on a one-dimensional regular lattice $\Lambda$, a Hamiltonian $H$ is \emph{geometrically local} if there exists a decomposition of $H = \sum_{i = 1}^n H_{i}$ such that $H_{i}$ is supported entirely on a neighborhood of diameter $\chi = O(1)$ including $i$.
To set the energy unit, we further assume that the spectral norm $\norm{H_i} \leq 1$ for every site $i$.
We call $\chi$ the \emph{interaction range} of $H$.
For example, $\chi = 1$ if $H$ is  a nearest-neighbor Hamiltonian.

In this work, we consider a Hamiltonian, $H = A + \alpha B$, where both $A$ and $B$ are geometrically local Hamiltonians.
Let $\chi$ be the interaction range of $A$.
Below, we illustrate our algorithm with time-independent $A$ and $B$ and later discuss the generalization to time-dependent Hamiltonians.

To set a baseline for later comparison, we briefly explain the complexity of simulating
$e^{-iHt}$ using the product formulas.
For that, we divide $[0,t]$ into $r$ equal time intervals so that $U(t) = U(t/r)^r$ and use the first-order product formula to approximate
\begin{align}
    U(t/r)\approx \prod_i e^{-i  A_{i}t/r } \prod_j e^{-i \alpha B_{j}t/r }.\label{eq:1st-trotter-1step}
\end{align}
The total error of this approximation can be upper bounded by
  $\O{ n t^2/r}$~\cite{childs_theory_2021},
which is independent of $\alpha$.
Given a fixed error tolerance, we choose $r \propto  n t^2 $.
In each interval, we use $\O{n}$ gates to implement the evolutions under $A_i$ and $B_i$, resulting in the total gate count
\begin{align}
    \mathcal G_{\text{Trotter,1}}  
    =\O { n^2 t^2}.\label{eq:gate-count-trotter-1}
\end{align}
Similarly, using a $p$th-order product formula would result in a gate complexity
\begin{align}
    \mathcal G_{\text{Trotter,p}}  
    =\O { (nt)^{1+1/p}},\label{eq:gate-count-trotter-p}
\end{align}
which approaches $\Otilde{ n t}$ in the large-$p$ limit.

Throughout the paper, we use $\Otilde{\cdot}$ to denote $\O{\cdot}$ up to subpolynomial corrections.
This gate count scales optimally with $n$ and $t$~\cite{haah2021quantum} and is independent of $\alpha$.
Meanwhile, one would expect that, if $A$ is fast-forwardable,  the gate count of simulating $H = A+\alpha B$ would scale proportionally with $\alpha$.

In the limit of small $\alpha$, $\alpha B$ represents a perturbation to the Hamiltonian $A$.
Naturally, we can move into the interaction frame of $A$ and write $U(t) = U_A(t) U_{BI}(t)$, where $U_A(t) = e^{-iAt}$ and 
$U_{BI}(t)$ is the evolution generated by the interaction-picture Hamiltonian
\begin{align}
    B_I(t) = \alpha U_A(t)^\dag B U_A(t).
\end{align}

Various quantum algorithms have been developed to simulate the interaction-picture evolution~\cite{low_hamiltonian_2019,berry_time-dependent_2020,an_time-dependent_2022,chenQuantumAlgorithmTimedependent2021new}.
Notably, Ref.~\cite{low_hamiltonian_2019} uses the truncated Dyson series to approximate $U_{BI}(t)$, resulting in a gate complexity $\Otilde {\alpha n^2 t}$ assuming the fast-forwardability of $A$.
Compared to high-order product formulas, this approach gains a favourable dependence on $\alpha$, but at the same time scales quadratically worst with $n$.
Another avenue to approximating $B_I(t)$ is the Magnus expansion~\cite{blanes2009magnus,arnal2018general}, which is defined as follows.
Let  $\alpha \norm{B} t \leq 1$. Then
there exists $\Omega(t) = \sum_{j = 1}^{\infty} \Omega_j(t)$ such that $U_{BI}(t) = e^{\Omega(t)}$, where $\Omega_j(t)$ is of the $j$th-order in $\alpha\norm{B}t$.
Truncating the expansion by keeping the first $q$ terms results in an approximation
\begin{align}
    U_{BI}(t) \approx \exp\left(\sum_{j = 1}^q \Omega_j(t)\right).
\end{align}
For example, at $q = 1$, we get $e^{\Omega_1(t)}$, where~$\Omega_1(t) = -i\int_0^t B_I(s)ds$ (See \cref{app:preliminaries} for a brief review).
Through this approximation, Ref.~\cite{an_time-dependent_2022} further approximates $\Omega_1(t)$ by a discrete sum, allowing $e^{\Omega_1(t)}$ to be implemented via a technique based on linear combinations of unitaries~\cite{childs_hamiltonian_nodate}.

Here, we propose to simply evaluate the truncated Magnus expansion, e.g., $\Omega_1(t)$, \emph{numerically}. 
This strategy has been previously explored in the context of simulating $2$-local Hamiltonians~\cite{m:chenQuantumSimulationHighlyoscillatory2023}. 
In our case, the interaction-picture Hamiltonian $B_I(t)$ at $t>0$ can be supported everywhere on the lattice, presenting a challenge for both the classical computation of the truncated Magnus expansion and its subsequent decomposition into elementary quantum gates. 
We overcome these challenges by showing that, for a small time step, the Lieb-Robinson bound guarantees an efficient Pauli decomposition of $\Omega_j(t)$ for all constant $j$.
In exchange for this mild increase in the classical computational overhead, our approach saves on additional quantum registers and achieves more favorable scaling than other approaches in certain regimes.

\section{Main results}

In this section, we illustrate the approximation of the time evolution under $H = A + \alpha B$ using the truncated Magnus expansion. Consider that 
\begin{align}
    e^{-i H t} &=  e^{-i A t} \mathcal{T}\exp\left( -i \alpha \int_0^t e^{i A s} B e^{-i As} ds \right) \nonumber\\
    &= e^{-i A t} e^{\Omega(t)}
    \approx e^{-i At}e^{\overline \Omega(t)},
\end{align}
where $\overline \Omega(t) \equiv \sum_{j = 1}^q \Omega_j(t)$.
For $t = O(1)$, we truncate terms outside the Lieb-Robinson lightcone to approximate $\overline \Omega(t)$ by a quasi-local operator  $\overline \Omega_\loc(t)$.
The locality of $\overline \Omega_\loc(t)$ allows us to efficiently estimate its Pauli decomposition.
Finally, we approximate the evolution $e^{\overline \Omega_\loc(t)}$ using a $p$th-order product formula.
In the following analysis, we use $V_{q,p}(t)$ to denote the resulting  unitary.  
For large $t$, we divide it into $r$ time intervals and repeat the decomposition for each interval as follows: 
\begin{align}
    U(t) = \left[e^{-iAt/r} V_{q,p}(t/r)\right]^r.\label{eq:moving_into_int_picture}
\end{align}

There are three contributions to the error of this simulation:
i) the truncation of the Magnus expansion, 
ii) the Lieb-Robinson light-cone approximation of $\overline \Omega$ by $\overline \Omega_\loc$,
iii) and the product formula for decomposing $e^{\overline \Omega_\loc}$.
We first consider the case $t = O(1)$ in the following analysis.

The Magnus expansion truncation error can be bounded using standard techniques, giving 
\begin{align}\label{eq:prem-magnus-bound}
    \Vert e^{\Omega(t)} - e^{\overline\Omega(t)}\Vert = \O{(\alpha n  t)^{q+1}}. 
\end{align}
Exploiting the locality of the Hamiltonian, we arrive at a tighter bound.

\begin{theorem}\label{thm:magnus-truncation-error}
Let $A$ be a geometrically local Hamiltonian and 
$B = \sum_{X\subset \Lambda} B_X$ be a $k$-local Hamiltonian such that $B_X$ is supported on a subset $X$ of at most $k$ qubits.
Let $d\equiv \max_{i \in \Lambda} \sum_{X: i \in X} \norm{B_X}$ denote the effective interaction degree of $B$.
Let $e^{\Omega(t)} = \mathcal{T}[ e^{ \int_0^t B_I(\tau) d\tau} ]$, where $B_I(\tau) = -i\alpha e^{i A \tau} B e^{-i A \tau}$. 
Let $\overline{\Omega}(t) = \sum_{j = 1}^q \Omega_j(t) $ be the truncated Magnus expansion of order $q$.
We have
\begin{align}
	\Vert e^{\Omega(t)} - e^{\overline{\Omega}(t)} \Vert \leq \mathcal{O}(n(\alpha d t)^{q+1})~. \label{eq:magnus-truncation-error} 
\end{align}
\end{theorem}
The theorem applies to a more general class of Hamiltonians $B$ and may be of independent interest beyond quantum simulation.  For the rest of the manuscript, we consider $B$ to be geometrically local. 
Our bound in Eq.~\eqref{eq:magnus-truncation-error} reduces the error dependence on the number of qubits from $\O{n^{q+1}}$ to $\O{n}$, a quadratic improvement even at~$q = 1$.
We present the detailed proof of \cref{thm:magnus-truncation-error} in \cref{app:magnus-truncation-error} and use the $q=1$ example to illustrate the idea.

If $\overline \Omega(t) = \Omega_1(t)$, the leading contribution to the truncation error comes from 
\begin{align}
    \norm{\Omega_2(t)} &\leq \frac{1}{2} \int_0^t d\tau \int_0^{\tau} ds\norm{\left[B_I(\tau), B_I(s)\right]}.\label{eq:example-omega-2}
\end{align}
The integrand can be bounded by~$\alpha^2\norm{B}^2 = \O{\alpha^2n^2}$, resulting in the bound in \cref{eq:prem-magnus-bound}.
Instead, if $B_I(\tau)$ and $B_I(s)$ \emph{were} both geometrically local, only Pauli terms in $B_I(\tau)$ of total norm $\O{1}$ may noncommute with a given term in $B_I(s)$.
Therefore, the commutator norm would scale as $\O{n}$, resulting in a bound of $\O{n(\alpha t)^2}$ for \cref{eq:example-omega-2}.

Although $B_I(\tau)$ and $B_I(s)$ are not strictly geometrically local due to the evolution under $A$, the geometric locality of $A$ and the Lieb-Robinson bound ensure that they inherit a similar structure. 
We develop a framework to capture the structure in \cref{app:concentrated-operator}, and show that the Magnus operators admit this structure in \cref{app:magnus-operator-structure}, and finally apply these results to bound the Magnus truncation error in \cref{app:magnus-truncation-error}.

Next, we use the Lieb-Robinson bound for the classical computation of $\overline \Omega(t)$. 
For $q = 1$, this operator can be expanded into a sum over $B_{X}$:
\begin{align}
    \overline \Omega(t) = \sum_{X} \int_0^t e^{iAs} B_{X} e^{-iAs} ds.
\end{align}
For each $X$, the evolved operator $e^{iAs} B_{X} e^{-iAs}$ is mostly supported within a light cone of radius $O(t) = O(1)$ around $X$.
In particular, let $A_{X,R}$ be the operator constructed from $A$ by dropping terms supported outside a radius $R$ from the set $X$.
The Lieb-Robinson bound implies~\cite{lieb_finite_2004,bravyi_lieb-robinson_2006}
\begin{align}
    \norm{e^{iAs} B_{X} e^{-iAs} - 
    e^{i A_{X,R}s} B_{X} e^{-i A_{X,R}s}
    } = \O{t e^{-R/\chi}},
\end{align}
where $\chi$ is the interaction range of $A$.
Therefore, the local version of the truncated Magnus operator
\begin{align}
    \overline \Omega_\loc (t) 
    \equiv \sum_{X}  \int_0^t e^{i A_{X,R}s} B_{X} e^{-i A_{X,R}s} ds
\end{align}
only differs from $\overline \Omega(t)$ by 
\begin{align}
    \norm{\overline \Omega(t) - \overline \Omega_\loc(t)} = \O{n t e^{-R/\chi }}. \label{eq:LR-error}
\end{align}
For a large time, we apply the construction in $r$ time steps, each of size $t/r$, resulting in the same bound. 
Choosing $R = \chi \log n t + \vartheta$ for some constant $\vartheta$ ensures that this error remains a small constant.
Note that $\overline \Omega_\loc(t)$ is a sum of terms that are supported on at most $\O{R} = \mathcal{O}(\log nt)$ qubits. Therefore, the complexity of computing $\overline \Omega_\loc(t)$ scales polynomially with~$n$ and $t$.
In \cref{app:magnus-operator-structure}, we generalize this construction to approximate $q$th-order Magnus operators by an operator $\overline \Omega_\loc(t)$ consisting of $n$ terms, each supported on at most $R + q$ qubits.
The error of this approximation is
\begin{align}
    \norm{\overline \Omega(t) - \overline \Omega_\loc(t)} = \O{n t e^{-R/8\chi }}, \label{eq:LR-error-general}
\end{align}
which is identical to \cref{eq:LR-error} except for a mildly slower decay.

Finally, we numerically find the Pauli decomposition $\overline \Omega_\loc(t)$ and approximate $e^{\overline \Omega_\loc(t)}$ using a $p$th-order product formula.
If $\alpha t = \O{1}$, the effective interaction degree of $\overline \Omega_\loc(t)$ is $\O{\alpha t}$ due to the first-order Magnus operator being the leading term (See~\cref{app:magnus-operator-structure}).
Therefore, the error of using the product formula is \cite{childs_theory_2021}
\begin{align}
    \norm{e^{\overline \Omega_\loc(t)}
    - V_{q,p}
    } = \O{n  (\alpha t)^{p+1}}. \label{eq:trotter-error}
\end{align}

Combining \cref{eq:magnus-truncation-error,eq:LR-error-general,eq:trotter-error}, the error of our algorithm for $t = O(1)$ is
\begin{align}
    \O{n(\alpha t)^{q+1}
    +nt e^{-R/8\chi}
    +n  (\alpha  t)^{p+1}}.
\end{align}
For large $t$, we repeat the approximation for $r$ time intervals, resulting in the total error
\begin{align}
    \O{\frac{n(\alpha  t)^{q+1}}{r^q}
    +nt e^{-R/8\chi}
    +\frac{n  (\alpha  t)^{p+1}}{r^p}}.
\end{align}
Given a constant error tolerance, we choose $R = 8\chi \log (nt) + \const$ and
\begin{align}
    r \propto \max\left\{
    n^{1/q} (\alpha  t)^{1+1/q},
    n^{1/p} (\alpha  t)^{1+1/p}
    \right\}.
\end{align}
At this point, it is obvious that the errors of the Magnus truncation and the product formula in our algorithm are comparable given the same approximation orders. 
To simplify the expression, we choose $p = q$ in the rest of the analysis.

To estimate the gate count, $\overline \Omega_\loc(t)$ contains $n$ terms, each supported on at most $R + q$ qubits. 
Since $R$ scales logarithmically with $nt$, the number of Pauli terms scales at worst polynomially with $n t$, i.e., there exists a constant $\gamma$ such that the number of Pauli terms is $\mathcal O(n (nt)^\gamma)$.
Here, $\gamma\leq 8\chi \log 4$ depends on the details of the lattice $\Lambda$ and the locality of the Hamiltonian $A$. 
Therefore, the total gate count for simulating $r$ instances of $V_{q,p}(t/r)$ is 
\begin{align}
    \mathcal G_{\text{Mag},p}^B
    =\Otilde {
    (nt)^\gamma (\alpha n t)^{1 + 1/p}
    }.\label{eq:gate-count-magnus-B}
\end{align}
Compared to that of the product formula [\cref{eq:gate-count-trotter-p}], the gate count of our algorithm scales nearly linearly with $\alpha$.
We summarize the result and compare to the state of the art in \cref{tab:gate-count}.

If $A$ is not fast-forwardable, we also have to simulate $r$ instances of $e^{-iAt/r}$ in addition to $V_{p,p}(t/r)$.
In the above analysis, we have not accounted for the error of simulating these evolutions.
To control such error, we further divide each interval of length $t/r$ into $r'$ equal subintervals and use the same $p$th-order product formula to approximate $e^{-iAt/(rr')}$.
The total error after $rr'$ intervals is
    $\O{n t^{p+1}/{(rr')^p}}$,
informing the choice of
\begin{align}
    r' \propto \frac{n^{1/p} t^{1+1/p}}{r} \propto \frac{1}{ \alpha ^{1+1/p}}.
\end{align}

The total gate count for simulating $r$ instances of $e^{-iAt/r}$ is given by 
\begin{align}
    \mathcal G_{\text{Mag},p}^A
    =
    \Otilde{\! (nt)^{1+1/p}\! }.\label{eq:magnus-gate-count-A}
\end{align}
Adding \cref{eq:magnus-gate-count-A} to \cref{eq:gate-count-magnus-B}, we arrive at the total gate count of our algorithm, which is asymptotically identical to that of the product formulas.
The advantage of using our algorithm is overshadowed by the cost of simulating $A$ in the asymptotic limit, but the effective reduction in the circuit depth can be significant for near-term simulations.

\begin{table}[t]
    \centering
    \begin{tabular}{c c c}
    \toprule
    Algorithm & Gate count \\\hline
    Product formulas~\cite{childs_theory_2021} & $\mathcal O ({(nt)^{1+o(1)}})$
    \\
    qDRIFT~\cite{campbellRandomCompilerFast2019new}     
    &  $\O{ \alpha^2 n^2 t^2 }$
    \\\hline
    Dyson series~\cite{low_hamiltonian_2019}     
    &  $\Otilde{ \alpha n^2  t}$
    \\
    qHOP~\cite{an_time-dependent_2022}     
    &  $\Otilde{ \alpha n^3 t^2 }$
    \\
    Cont. qDRIFT~\cite{berry_time-dependent_2020} 
    & $\Otilde{\alpha^2 n^2t^2}$ 
    \\\hline
    This paper
    & $\Otilde {
        (nt)^\gamma (\alpha n t)^{1 + o(1)}
        }$
    \\
    \botrule
    \end{tabular}
    \caption{The gate count of several algorithms for simulating the evolution $e^{-iHt}$ generated by the Hamiltonian $H = A + \alpha B$, as considered in this paper. Except for the $p$th-order product formulas and qDRIFT, these algorithms are applied to the evolution $U_{BI}(t)$ in the interaction frame. It is assumed that the cost of implementing $e^{-i A t}$ is negligible for all $t$. In the scenario where $\gamma$ approaches 0, our algorithm's gate count scales as $\Otilde{\alpha n t}$, thereby outperforming the other algorithms.
    }
    \label{tab:gate-count}
\end{table}

To benchmark our algorithm on near-term simulations, we consider the disordered XY model $H = A + \alpha B$ with $A = \sum_{i=1}^{n} Z_i$ and $B = \sum_{i=1}^{n-1} r_i Y_i Y_{i+1}+ s_i X_i X_{i+1} - \sum_{i=1}^n u_i X_i $ for $n = 12$ qubits. Here, $r_i, s_i$, and $u_i$ are sampled uniformly at random in the range $(-1, 1)$ for all $i$. 
In~\cref{fig:error_plots}, we plot the approximation error in simulating $e^{-i Ht}$ for a total time $t =3$ at different $\alpha$ using our algorithm, the first-order, and the second-order product formulas.  
To ensure the three algorithms have the same gate count, we use time intervals of different sizes (1, 0.5, and 0.75 respectively) for different algorithms.

\Cref{fig:error_plots} shows that for small values of $\alpha$, our algorithm performs significantly (one-order magnitude) better than the first- and the second-order product formulas. The dependence on $\alpha$ of the algorithms also matches the theoretical prediction. In particular, errors due to the first- and second-order product formulas scale linearly in $\alpha$, whereas the error of our algorithm scales quadratically. 

\begin{figure}[b]
    \centering
    \includegraphics[width=0.45\textwidth]{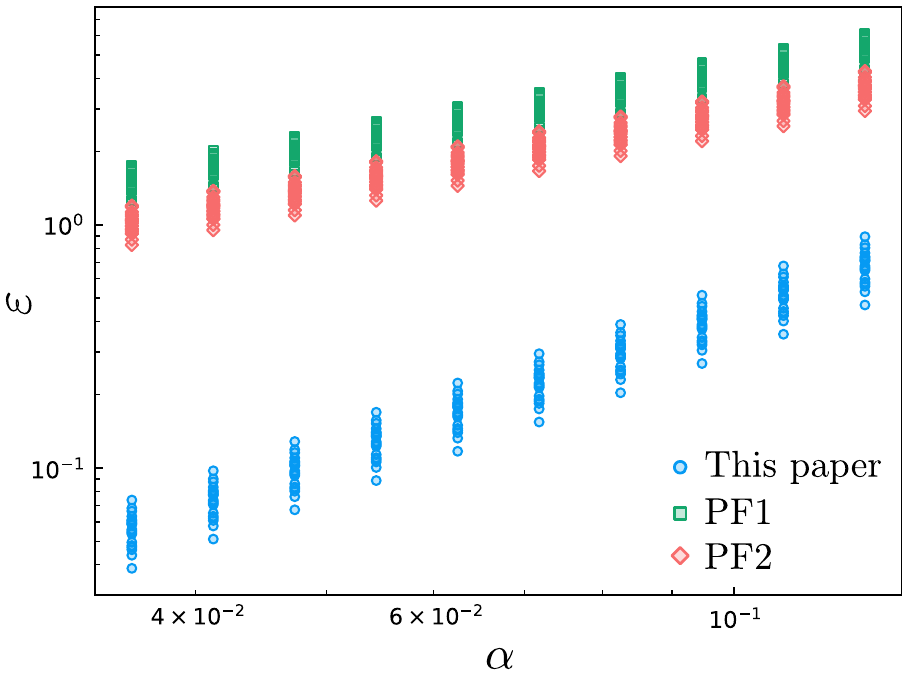}
    \caption{Comparison of the errors from our algorithm (circles), first- (``PF1", squares) and second-order (``PF2", diamonds) product formulas for simulating the disordered XY model of $n = 12$ qubits for time $t = 3$ at different values of $\alpha$. 
    Points of the same color represent different random coefficients of the Hamiltonian.
    The errors for the first- and second-order product formulas scale linearly with $\alpha$, whereas the error of our algorithm scales quadratically.}
    \label{fig:error_plots}
\end{figure}

\section{Discussion}

We have introduced an algorithm for simulating the time evolution in the interaction frame of a geometrically local Hamiltonian.
The algorithm is ancilla-free and takes advantage of the locality of the Hamiltonian, particularly via a tightened bound on the error of truncating the Magnus expansion.

In the limit~$\chi\!\rightarrow 0$, $A$ contains only terms supported on single qubits. The evolution $e^{-iAt}$ is fast-forwardable and the resulting lightcone does not grow with time~\!($\gamma\!\!\!\rightarrow \!0$).
The complexity of our algorithm reduces to $\Otilde{(\alpha n t)^{1+o(1)}}$\!, which is optimal~\cite{haah2021quantum} and outperforms both~$\Otilde{\alpha^{o(1)}(nt)^{1+o(1)}}$\! of product formulas and~$\Otilde{\alpha n^2 t}$\! of algorithms based on the Dyson series.

In general, the complexity of our algorithm depends on the constant $\gamma$, which can be as large as $4\chi\log 4$.
This value assumes that $\overline \Omega_\text{loc}(t)$ contains all exponentially many possible Pauli strings supported inside the lightcone, preventing applications of our algorithm to more general $B$, e.g. $k$-local Hamiltonians.
Practically, $\gamma$ depends on the Hamiltonian of interest and can be much smaller than this upper bound, especially in the presence of symmetries.

For example, consider $A\propto \sum_i X_iX_{i+1} + Y_iY_{i+1}$ and $B$ contains Pauli strings of the form $\sigma_i Z_{i+1} \dots Z_{j-1} \sigma_j$, where $\sigma_i,\sigma_j$ can be either $X,$ or $ Y$ supported on qubits $i,j$.
The evolution of $B$ under $A$ preserves this structure, limiting the number of possible Pauli strings in $\overline{\Omega}_\text{loc}$ to scale polynomially with the size of the lightcone, which is logarithmic in $n$ and $t$. 
In this case, $\gamma$ can be chosen arbitrarily close to 0.
Such $A$ and $B$ are the products of applying the Jordan-Wigner transformation on a nearest-neighbor and a two-local fermionic Hamiltonian, respectively. 
The argument for $\gamma\approx 0$ can be similarly generalized to $k$-local fermionic Hamiltonians $B$.

If $\Lambda$ is a general $D$-dimensional lattice, the number of Pauli strings in the decomposition of the Magnus operators can be exponential in $R^D$, which is no longer polynomial in $n,t$ for $D> 1$. This number again corresponds to the worst case and can be significantly smaller, e.g. under the symmetry constraints discussed earlier. For such $D$-dimensional systems, our algorithm still applies.

While we assume $A$ and $B$ are time-independent in our analysis,
our algorithm can be immediately generalized to simulate time-dependent Hamiltonians $A(t)$ and $B(t)$.
In the special case $A(t) = 0$, our algorithm can be applied to simulate a time-dependent geometrically local Hamiltonian $B(t)$.
Since the Magnus operators are computed classically, the gate count of our algorithm is independent of how fast $B(t)$ changes with time---a feature similar to continuous qDRIFT~\cite{berry_time-dependent_2020}.
For reference, the complexity of product formulas, qHOP~\cite{an_time-dependent_2022}, and algorithms based on the Dyson series~\cite{low_hamiltonian_2019} would scale either polynomially or logarithmically with the derivative $\norm{dB/dt}$.

{\ }

\textit{Note added.} While finalizing our paper, we noticed two independent papers on arXiv studying the Hamiltonian simulation using the Magnus expansion~\cite{bosse_efficient_2024, casares_quantum_2024}.

\begin{acknowledgments}
    We acknowledge helpful discussions with Sergey Bravyi, William Kirby, Antonio Mezzacapo, Oles Shtanko, and Kristan Temme. 
\end{acknowledgments}

\onecolumngrid

\appendix

\vspace{0.5in}

\setcounter{theorem}{0}

\begin{center}
	{\large \textbf{Appendix for ``Hamiltonian Simulation in the Interaction Picture using the Magnus Expansion''} }
\end{center}

\section{Preliminaries}\label{app:preliminaries}

In this section, we recall definition and prior results relevant for our paper. We first define the notion of Hamiltonian simulation in the interaction picture. We then summarize properties and relevant formulas for the Magnus expansion. For completeness, we also summarize Duhamel's principle for Hamiltonian simulation and commutator-dependent bounds for first-order and second-order Trotter simulation. We point readers to \cite{lin2022lecture} for more details on Hamiltonian simulation methods. 

Let 
\begin{align}\label{eq:ham}
 H = A + \alpha B   
\end{align}
where $A$ is geometrically local Hamiltonians and $B$ is a $k$-local Hamiltonian with an effective interaction degree $d$. Let $B = \sum_{X \subset \Lambda} B_X$ be a decomposition into a sum of $L$ terms, each supported on a subset $X$ of at most $k$ qubits.
We call 
\begin{align}
d\equiv \max_{i \in \Lambda} \sum_{X: i \in X} \norm{B_X}    
\end{align} 
the effective interaction degree of $B$, which upper bounds the norms of the terms supported on a qubit. 
For example, $d = \mathcal O{(1)}$ if B is also geometrically local and $d = \mathcal O{(n^{k-1})}$ in general.

We recall the definition of geometrical locality as defined in \cref{sec:setup}. In particular,
a Hamiltonian $H$ for $n$ qubits defined on a $1$-dimensional regular lattice is geometrically local if there is a decomposition $H = \sum_{i=1}^n H_i$ such that $H_i$ is supported entirely on a neighborhood of constant radius around $i$. In our analysis, we assume that $\alpha <1$.

In the interaction frame of $A$, the time evolution under $H$ can be expressed as 
\begin{align}
	e^{-i H t} = e^{-i A t}  \mathcal{T}\left[e^{ \int_0^t B_I(\tau) \mathrm{d} \tau}\right]~, 
\end{align}
where $ \mathcal{T}\left[e^{ \int_0^t B_I(\tau) \mathrm{d} \tau}\right]$ is the time-ordered expeonential with 
\begin{align}\label{eq:time-evolved-B}
	B_I(\tau) = - i \alpha e^{i A\tau } B e^{-i A \tau},
\end{align}
with $\norm{B_I(\tau)} = \norm{B} \leq nd$ which follows trivially from the definition of $d$ and a sum over the qubit index. 

Before providing an approximation of  $\mathcal{T}\left[e^{\int_0^t B_I(\tau) \mathrm{d} \tau}\right]$, we summarize important properties of the Magnus expansion. The Magnus expansion is defined as follows.

\begin{lemma}[Magnus expansion \cite{moan2008convergence,blanes2009magnus,arnal2018general}]\label{lemma:magnus-expansion}
Let $\mathcal{A}(\tau)$ be a continuous operator-valued function defined for $0\leq \tau\leq t $ such that $\Vert \mathcal{A}(\tau) \Vert \leq 1$. Then, we get that 
\begin{align}
	\mathcal{T} \exp \left\{\int_0^t \mathrm{~d} \tau \mathcal{A}(\tau)\right\}=\exp \left(\sum_{j=1}^{\infty} \Omega_j(t)\right)\equiv \exp(\Omega(t))~,
\end{align}
where 
\begin{align}
	&\Omega_j(t)=\frac{1}{j^2} \sum_\sigma(-1)^{d_b} \frac{1}{\left(\begin{array}{c}
j-1 \\
d_b
\end{array}\right)} \int_0^t d \tau_1 \ldots \int_0^{d \tau_{j-1}} d \tau_j\left[\mathcal{A}\left(\tau_1\right), \ldots\left[\mathcal{A}\left(\tau_{j-1}\right), \mathcal{A}\left(\tau_j\right)\right] \ldots\right], 
\end{align} 
the sum is being taken over all permutations $\sigma$ of $\{1, \dots, j\}$, and $d_b$ is the number of descents, i.e., pairs of consecutive numbers $\sigma_k, \sigma_{k+1}$ for $k =\{1, \dots, j-1\}$ such that $\sigma_k > \sigma_{k+1}$ in the permutation. 
\end{lemma}

From Lemma~\ref{lemma:magnus-expansion}, the first two $\Omega_j(t)$ is given by: 
\begin{align}
	 \Omega_1(t) &=  \int_0^t \mathrm{~d}\tau \mathcal{A}(\tau)~,\\
	\Omega_2(t) &= \frac{1}{2} \int_0^t \mathrm{~d} \tau_1 \int_0^{\tau_1} \mathrm{~d} \tau_2\left[\mathcal{A}(\tau_1), \mathcal{A}(\tau_2)\right]~.
\end{align}
For the problem of interest in this paper, $\mathcal{A}(\tau)$ gets replaced by $B_I(\tau)$ as defined in Eq.~\eqref{eq:time-evolved-B} for the Hamiltonian considered in Eq.\eqref{eq:ham}. 

\medskip 

We now recall an alternate representation for the Magnus expansion which we employ in our proofs. In particular, we will invoke the following lemma in Section~\ref{app:magnus-truncation-error}. 
\begin{lemma}[Magnus expansion \cite{magnus1954exponential}]
	Let $ Y^{\prime}(t) = A(t)Y(t)$ with $Y(0)=I$. Then $Y(t)$ can be written as follows:
	\begin{align}
		Y(t) = \exp(\Omega(t))~,
	\end{align}
	where
	\begin{align}
	\frac{d \Omega}{d t}=\sum_{n=0}^{\infty} \frac{B_n}{n !} \operatorname{ad}_{\Omega}^n ~,
	\end{align}	
	with $B_n$ denoting the Bernoulli numbers.  Then the integration over $d\Omega/dt$ leads to 
	\begin{align}
		\Omega(t)=\int_0^t A\left(t_1\right) d t_1-\frac{1}{2} \int_0^t\left[\int_0^{t_1} A\left(t_2\right) d t_2, A\left(t_1\right)\right] d t_1+\cdots~.
	\end{align}
	Moreover, derivatives of $\Omega_j$ have the following form: 
	\begin{align}
		\Omega_1^{\prime} &= A~,\\
		\Omega_2^{\prime}& =-\frac{1}{2}[\Omega_1, A]~,\\
		\Omega_3^{\prime}& = -\frac{1}{2}[\Omega_2, A]+\frac{1}{12}[\Omega_1,[\Omega_1, A]]~,\\
		\Omega_4^{\prime}& = -\frac{1}{2}[\Omega_3, A]+\frac{1}{12}[\Omega_2,[\Omega_1, A]]+\frac{1}{12}[\Omega_1,[\Omega_2, A]]~,\\
		\Omega_j^{\prime} & = \sum_{k=1}^{j-1} \frac{B_k}{k!} S_j^{(k)}, \quad j\geq 2~,
	\end{align}
	where
	\begin{align}
		S_j^{(k) }= \sum \left[ [\Omega_1, [\dots[\Omega_{i_k}, A]\dots]  \right] \qquad (i_1 + \dots + i_k = j-1)~.
	\end{align}
\end{lemma}

\bigskip 

\noindent \textbf{Duhalmel's principle}. We recall Duhamel's principle for Hamiltonian simulation which is useful in the error analysis. Let $U(t), V(t) \in \mathbb{C}^{d\times d}$. Let $H\in \mathbb{C}^{d\times d}$ denote a Hermitian operator and $E(t) \in \mathbb{C}^{d\times d}$ denote an arbitrary matrix. If 
\begin{align}
    i \partial_t U(t) &= H U (t)~, \\
    i \partial_t V(t) & = H V(t)+E(t)~,\\
    U(0) & = V(0) = I~,
\end{align}
then 
\begin{align}
    V(t) = U(t) - i \int_0^t U(t-s) E(s) ds,
\end{align}
which implies that 
\begin{align}\label{eq:duhamel}
\Vert U(t)- V(t)\Vert \leq \int_0^t ds \Vert E(s)\Vert~.
\end{align}
Here, we used the data-processing inequality and the triangle inequality.

\bigskip 

\noindent \textbf{Product formulas and Trotter error}. Using the Duhamel's principle stated above, one can derive Trotter error bounds that depends on the nested commutators between the terms of a Hamiltonian $H$. In particular, let $H = \sum_{i=1}^L H_i$. Let $V_p$ denote the $p$-th order product formula, which is defined recursively as follows: 
\begin{align}
    V_1(t) &= \prod_{j=1}^L \exp\left(-i H_j t\right)~,\\
    V_2(t) &= \prod_{j=1}^L \exp\left(-i H_j/2\right) \prod_{j=L}^1 \exp\left(-i H_j/2\right)~,\\
    V_p(t) &= V_{2p-2}(m_p t)^2 V_{2p -2}( (1- 4m_p) t)V_{2p -2}(m_p t)^2~,
\end{align}
where $m_p = \frac{1}{4 - 4^{1/(2p-1)}}$.

The number of $r$ Trotter steps for a $p$-th order product formula implies $(V_p(t/r))^r$, i.e., one implements $V_p(t/r)$ for time $t/r$ and repeats that $r$ times.  

 Let $V_1$ and $V_2$ denote the first-order and second-order product formula, respectively, and let $U = e^{-i H t}$. Then we get that \cite{childs2022quantum}
 \begin{align}
     \Vert U - V_1 \Vert  &\leq \frac{t^2}{2} \sum_{k = 1}^L \left\Vert \left[\sum_{j=k+1}^L H_j, H_k \right]\right\Vert~ \\
     \Vert U - V_2 \Vert & \leq \frac{t^3}{12}\sum_{i=1}^L \left\Vert \left[\sum_{k = i+1}^L H_k,\left[\sum_{j=i+1}^L H_j, H_i\right]\right]\right\Vert + \frac{t^3}{24} \sum_{i = 1}^L \left\Vert \left[H_i, \left[H_i, \sum_{j=i+1}^L H_j \right] \right] \right\Vert 
 \end{align}

For the case when $H = A + \alpha B$, then we have
\begin{align}
    \Vert U - V_1 \Vert  &\leq \frac{ \alpha t^2}{2} \Vert [A,B]\Vert~,\\
    \Vert U - V_2 \Vert & \leq \alpha\left(\frac{ t^3}{12} \Vert [A [A, B] ]  \Vert + \frac{\alpha t^3}{24} \Vert [B,[B,A]] \Vert\right)~.
\end{align}

\bigskip 

\noindent \textbf{Lieb-Robinson bound}. Let $O_X$ and $O_Y$ operators on $n$ qubits supported on $X$ and $Y$ sets of qubits, and $H$ denote a geometrically local Hamiltonian with an interaction range $\chi$. Then the following inequality holds~\cite{lieb_finite_2004}: 
\begin{align}
    \| [e^{i H t} O_Y e^{-i Ht}, O_X] \| \leq c  \min\{|X|, |Y|\}  \| O_X\| \| O_Y \| \left(e^{v\abs{t}}-1\right) \exp\left(-\frac{R}{\chi}\right) \label{eq:lb-bound},
\end{align}
where $R$ is the minimum distance between $X$ and $Y$, i.e., $\min_{i\in X, j\in Y}  d_{ij}$, and $c$ and $v$ are constants that depend on $\chi$ and the effective interaction degree of $H$ which is also a constant. 

In the next section, we derive a naive bound on the truncated Magnus approximation. 

\section{Proof of Eq.~\eqref{eq:prem-magnus-bound}}
In this section we establish a bound on the truncation error in the Magnus expansion. In particular, we highlight why exploiting the commutator structure is necessary to establish a tight bound. As a first step,  we approximate $\Omega$ with $\overline{\Omega} =\sum_{k=1}^q \Omega_k $.
Note that in our problem, $\Omega_k$ are anti-Hermitian. For simplicity, we do not write the time dependence of $\Omega$ explicitly. 

Since $\norm{\Omega_k} \leq 2^k \alpha^k \norm{B}^k t^k \leq (2n\alpha d t)^k$, we have 
\begin{align}
 \Vert e^{\Omega} - e^{\overline{\Omega}}\Vert 
	& \leq \Vert \Omega - \overline{\Omega}\Vert 
	 = \sum_{k=q+1}^{\infty} \Vert \Omega_k\Vert
     \leq \sum_{k=q+1}^{\infty} (2n\alpha d t)^k
     \leq 2 (2n\alpha d t)^{q+1},\label{eq:naive-bound-on-Magnus-error}
\end{align}
assuming $n\alpha dt \leq 1/4$.
The first inequality follows from the Duhamel's principle as discussed in Eq.~\eqref{eq:duhamel}. We recall it for completeness. Let $\tau$ denote a dummy variable. Note that $\partial_{\tau}(e^{\Omega \tau}) = \Omega e^{\Omega \tau}$ and $\partial_{\tau}e^{\overline{\Omega}\tau} = \Omega e^{\overline{\Omega} \tau} + (\overline{\Omega}- \Omega)e^{\overline{\Omega}\tau} $. Then we get 
\begin{align}
	e^{\overline{\Omega} \tau} &= e^{\Omega \tau} +  \int_0^{1} e^{\Omega (\tau -s)}  (\overline{\Omega}-\Omega) e^{\overline{\Omega}\tau} d\tau \label{eq:duhamel-app1}
\end{align}
which implies that 
\begin{align}
	\Vert e^{\Omega } - e^{\overline{\Omega}  }\Vert &= \left\Vert \int_0^{1} d\tau  e^{\Omega (\tau -s)} (\Omega - \overline{\Omega}) e^{\overline{\Omega}\tau}  \right\Vert 
	\leq \int_0^{1} d\tau \left\Vert \Omega - \overline{\Omega}\right\Vert 
	  = \Vert \Omega - \overline{\Omega}\Vert~.\label{eq:duhamel-app2}
\end{align}

We conclude from \cref{eq:naive-bound-on-Magnus-error} that the simulation error is of order $\alpha^{q+1}$. However, the $n$ dependence makes the bound loose as it grows with $n^{q+1}$. In the next sections, we show how to improve the bound to a linear dependence on $n$ by exploiting the locality $B_I(\tau)$.

\section{Magnus truncation error}\label{app:magnus-truncation-error}
In this section, we prove \cref{thm:magnus-truncation-error} by analyzing the simulation error that takes account of the commutator between different $\Omega_i$. 
First, we note that
\begin{align} 
    \norm{e^{-\overline \Omega}e^{\Omega} - I}     
    = \norm{\int_0^1 e^{-\overline \Omega \tau}(\Omega- \overline\Omega)e^{\Omega\tau} d\tau}
    \leq \int_0^1 \norm{\Omega- \overline\Omega} d\tau
    \leq 2 (2n\alpha d t)^{q+1},\label{eq:order-condition}
\end{align}
where the last inequality follows from \cref{eq:naive-bound-on-Magnus-error}.
\Cref{eq:order-condition} implies that $\partial_t (e^{-\overline \Omega}e^{\Omega}) = \O{\alpha^{q+1}}$ and, specifically,
\begin{align} 
    \partial_\alpha^{k} \partial_t (e^{-\overline \Omega}e^{\Omega}) = 0
\end{align}
for all $k \leq q$.

In the following discussion, we shall prove a bound on $\norm{\partial_t (e^{-\overline \Omega}e^{\Omega})}$ that scales linearly with $n$ and use that to bound the Magnus truncation error.

The derivative of  $e^{\Omega}$ is $\partial_{t} e^{\Omega}  = B_I e^{\Omega }$. Recall that for simplifying the notation, we did not write the dependence of $\Omega$ on $t$ explicitly. For $\wt{\Omega}(t)$ we write the $t$ dependence and compute the derivative of $e^{-\wt{\Omega}(t)}$ as follows: 
\begin{align}
	\partial_{t} e^{-\wt{\Omega}(t)} & = \lim_{dt \to 0}\frac{e^{-\wt{\Omega}(t+dt)} - e^{-\wt{\Omega}(t) }}{dt}\\
	& = \lim_{dt \to 0} \frac{e^{- \wt{\Omega} (t)  - \wt{\Omega}(t)^{\prime} dt  + O(dt^2) } - e^{- \wt{\Omega} (t)} }{dt}\\
	& = \lim_{dt \to 0} \lim_{\tau \to 1}\frac{e^{- \wt{\Omega }(t)\tau} \left( \T \left[  \exp \left(-  i \int_0^{\tau} dx   e^{\wt{\Omega}(t)x} (\wt{\Omega}(t)^{\prime} dt)  e^{-\wt{\Omega}(t) x}      \right)\right]  - I  \right)}{dt} \\
	& = - e^{- \wt{\Omega }(t)}\left(\int_0^{1} dx e^{ \wt{\Omega}(t)x } \wt{\Omega}(t)^{\prime} e^{- \wt{\Omega}(t)x} \right)~,
\end{align}	
which follows from the first-order term in the expansion of the time-ordered exponential, and also by taking the limit $dt\to 0$.  
Using this identity, we have
\begin{align}
\partial_{t} (e^{-\overline{\Omega} } e^{\Omega}) &=( \partial_{t} e^{-\overline{\Omega} }) e^{\Omega} + e^{-\overline{\Omega} }\partial_{t} e^{\Omega}\\
&=  - e^{- \wt{\Omega }(t)}\left(\int_0^{1} dx e^{ \wt{\Omega}(t)x } \wt{\Omega}(t)^{\prime} e^{- \wt{\Omega}(t)x} \right) e^{\Omega}+\alpha e^{-\overline{\Omega}} B_Ie^{\Omega}\\
& = - e^{- \wt{\Omega }(t)}\left( \sum_{j=1}^q\int_0^{1} dx e^{ \wt{\Omega}(t)x } {\Omega}_j(t)^{\prime} e^{- \wt{\Omega}(t)x} - \int_0^1 dx \Omega_1(t)^{\prime}   \right)e^{\Omega(t)},
\end{align}
where we used the fact that $e^{\Omega} = \mathcal{T}\left[e^{\int_0^t B_I(\tau)  \mathrm{d} \tau}\right]$ and therefore, $\partial_{t} e^{\Omega} =B_I(t) e^{\Omega}$. The second term in the last equality follows from the fact that $B_I(t) = \Omega_1(t)'$ and we trivially added an integration over $x$. 

Terms inside the bracket in the aforementioned equation can be simplified as follows:
\begin{align}
	&\sum_{j=1}^q\int_0^{1} dx e^{ \wt{\Omega}(t)x } {\Omega}_j(t)^{\prime} e^{- \wt{\Omega}(t)x} - \int_0^1 dx \Omega_1(t)^{\prime}  \nonumber \\
	& = \sum_{j=2}^q\int_0^{1} dx e^{ \wt{\Omega}(t)x } {\Omega}_j(t)^{\prime} e^{- \wt{\Omega}(t)x} + \left( \int_0^{1} dx e^{ \wt{\Omega}(t)x } {\Omega}_1(t)^{\prime} e^{- \wt{\Omega}(t)x}- \int_0^1 dx \Omega_1(t)^{\prime}  \right)\\
	& = \sum_{j=2}^q \int_0^{1} dx e^{ \wt{\Omega}x } {\Omega}_j^{\prime} e^{- \wt{\Omega}x} + \int_0^1 dx \int_0^x dy e^{\wt{\Omega}y} [ \wt{\Omega}, \Omega_1^{'}]e^{-\wt{\Omega}y}~.\label{eq:C11}
\end{align}

Here, we used the following argument:
\begin{align}
	e^{\wt{\Omega}x} \Omega_1^{'} e^{-\wt{\Omega}x} - \Omega_1^{'} &= \int_0^x dy \partial_y e^{\wt{\Omega}y} \Omega_1^{'} e^{-\wt{\Omega}y}  
	= \int_0^x dy e^{\wt{\Omega}y} [ \wt{\Omega}, \Omega_1^{'}]e^{-\wt{\Omega}y}
 \end{align}

 We analyze $ \int_0^x dy e^{\wt{\Omega}y} [ \wt{\Omega}, \Omega_1^{'}]e^{-\wt{\Omega}y}$ further by adding and subtracting $\int_0^x dy [ \wt{\Omega}, \Omega_1^{'}]$. Then 
 \begin{align}
\int_0^x dy e^{\wt{\Omega}y} [ \wt{\Omega}, \Omega_1^{'}]e^{-\wt{\Omega}y} &=   \int_0^x dy e^{\wt{\Omega}y} [ \wt{\Omega}, \Omega_1^{'}]e^{-\wt{\Omega}y}  - \int_0^x dy [ \wt{\Omega}, \Omega_1^{'}] + \int_0^x dy [ \wt{\Omega}, \Omega_1^{'}]\\
& = \int_0^x dy \int_0^y dz e^{\wt{\Omega}z} [\wt{\Omega},[\wt{\Omega},\Omega_1^{'}]] e^{-\wt{\Omega}z} + \int_0^x dy [ \wt{\Omega}, \Omega_1^{'}]~.
 \end{align}

 By repeating the same procedure multiple times, we get that 
 \begin{align}
  \int_0^x dy e^{\wt{\Omega}y} [ \wt{\Omega}, \Omega_1^{'}]e^{-\wt{\Omega}y} &  = \underbrace{\int_0^x dv_q\cdots \int_0^{v_2} dv_1}_{q~\operatorname{times}} e^{\wt{\Omega}v_1} \underbrace{[\wt{\Omega}, [\dots, [\wt{\Omega}, \Omega_1^{\prime}]]\dots]}_{q~\operatorname{nested ~ commutators}}e^{-\wt{\Omega}v_1} \nonumber \\
  & \qquad \qquad +\sum_{j=1}^{q-1} \int_0^x dv_{j} \cdots \int_0^{v_2} dv_1  \underbrace{[\wt{\Omega}, [\dots, [\wt{\Omega},  \Omega_1^{\prime}]]\dots]}_{j~\operatorname{nested ~ commutators}}\label{eq:approx3}\\ 
  &= \int_0^x dv_q\cdots \int_0^{v_2} dv_1 e^{\wt{\Omega}v_1}\operatorname{adj}^q_{\wt{\Omega}}\Omega_1^{\prime}e^{-\wt{\Omega}v_1}  +\sum_{j=1}^{q-1} \int_0^x dv_{j} \cdots \int_0^{v_2} dv_1  \operatorname{adj}^j_{\wt{\Omega}}\Omega_1^{\prime}~\\
    &= \int_0^x dv_q\cdots \int_0^{v_2} dv_1 e^{\wt{\Omega}v_1}\operatorname{adj}^q_{\wt{\Omega}}\Omega_1^{\prime}e^{-\wt{\Omega}v_1}  +\mathcal P_q(\alpha)~. \label{eq:C17} 
 \end{align}
 where we use $\mathcal P_q(\alpha)$ to denote polynomials of degree at most $q$ in $\alpha$.
 Using similar techniques, we have 
 \begin{align} 
      \int_0^{1} dx e^{ \wt{\Omega}x } {\Omega}_j^{\prime} e^{- \wt{\Omega}x}
      = \int_0^1 dx \int_0^x dv_{q-j+1}\dots \int_0^{v_2} dv_{1}
      e^{\wt{\Omega}v_1}
       \operatorname{adj}_{\overline \Omega}^{q-j+1} \Omega_j^{\prime}
       e^{-\wt{\Omega}v_1} + \mathcal P_q(\alpha). \label{eq:C18}
 \end{align}
 Substituting \cref{eq:C17,eq:C18} into \cref{eq:C11}, we get
 \begin{align}
 \partial_t (e^{-\overline{\Omega} } e^{\Omega})
 = \underbrace{\sum_{j = 1}^q \int_0^1 dx \int_0^x dv_{q-j+1}\dots \int_0^{v_2} dv_{1}
       e^{\wt{\Omega}v_1}
        \operatorname{adj}_{\overline \Omega}^{q-j+1} \Omega_j^{\prime}
        e^{-\wt{\Omega}v_1}}_{=\O{\alpha^{q+1}}} + \mathcal P_q(\alpha). \label{eq:before-P-vanish}
 \end{align}
 Recall that $ \partial_t (e^{-\overline{\Omega} } e^{\Omega}) = \O{\alpha^{q+1}}$.
 Therefore, $\mathcal P_q(\alpha)$ must vanish in \cref{eq:before-P-vanish}, leaving
 \begin{align} 
 \partial_t (e^{-\overline{\Omega} } e^{\Omega})
 = \sum_{j = 1}^q \int_0^1 dx \int_0^x dv_{q-j+1}\dots \int_0^{v_2} dv_{1}
       e^{\wt{\Omega}v_1}
        \operatorname{adj}_{\overline \Omega}^{q-j+1} \Omega_j^{\prime}
        e^{-\wt{\Omega}v_1}. \label{eq:after-P-vanish}
 \end{align}

 We can now use \cref{eq:after-P-vanish} to bound the Magnus truncation error:
 \begin{align} 
      \norm{e^{\overline{\Omega}}- e^{\Omega}}
      &\leq \norm{e^{-\overline{\Omega}}e^{\Omega} - I}
      \leq \int_0^t dt \norm{\partial_t (e^{-\overline{\Omega}}e^{\Omega})}\\
      &\leq \sum_{j = 1}^q \int_0^t dt  \int_0^1 dx \int_0^x dv_{q-j+1}\dots \int_0^{v_2} dv_{1}
      \norm{
              \operatorname{adj}_{\overline \Omega}^{q-j+1} \Omega_j^{\prime}}\\
       &= \sum_{j = 1}^q   \frac{t}{(q-j+2)!}
      \norm{
              \operatorname{adj}_{\overline \Omega}^{q-j+1} \Omega_j^{\prime}}.
 \end{align}
Using \cref{cor:Omega-structure,lem:commutator-between-H-K}, we have that
\begin{align} 
\operatorname{adj}_{\overline \Omega}^{q-j+1} \Omega_j^{\prime}
          \in \mathcal H(c_j,\kappa_j,k_j,\alpha^{q+1} d^{q+1} t^{q})
\end{align}
for some constants $c_j,\kappa_j, k_j$ that depend on $q, k$ but are independent of $n, \alpha, d, t$.
Applying \cref{lem:norm-of-concentrated-H} to bound their norms, we arrive at
\begin{align} 
     \norm{e^{\overline{\Omega}}- e^{\Omega}} = \O{n (\alpha t d)^{q+1}}.
\end{align}

\section{Concentrated $k$-local Hamiltonians}\label{app:concentrated-operator}

Under the evolution generated by $A$, the Hamiltonian $B$ is no longer $k$-local.
However, because $A$ is geometrically local, the interaction-picture Hamiltonian and the Magnus operator inherits many properties of a $k$-local Hamiltonian.
In this section, we develop tools to capture these properties.

\begin{definition}[Concentrated Operators]\label{def:concentrated-operator}
An operator $A$ is considered concentrated on site $i$ if it admits a decomposition $A = \sum_{r = 0}^{\infty} A_{i,r}$ such that
\begin{itemize}
     \item $A_{i,r}$ is supported entirely within a radius $r$ from $i$, and
     \item there exists a constant $c$ and $\kappa$ such that $\norm{A_{i,r}} \leq c e^{-r/\kappa}$ for all $i, r$. 
 \end{itemize} 
 We use $\mathcal C_i(c,\kappa)$ to denote the set of such operators $A$ and call $i$ the center of $A$.
\end{definition}

If $A \in \mathcal C_i(c,\kappa)$, it follows that $\norm{A} \leq c(1+\kappa)$.
Note that if $A \in \mathcal C_i(c,\kappa)$, then $A \in \mathcal C_i(c',\kappa')$ for all $c' > c, \kappa' > \kappa$.

The following lemma provides a structure for the commutator between two concentrated operators.
\begin{lemma}\label{lem:comm-decomposition}
If $A \in \mathcal C_{i}(c,\kappa)$ and $B \in \mathcal C_{j}(c,\kappa)$, there exists a set of concentrated operators $\{\hat{A}_\ell: \hat{A}_\ell \in C_{i}(c,2\kappa) \}$ 
and $\{ \hat{B}_\ell: \hat{B}_\ell \in C_{j}(c,2\kappa) \}$ such that
\begin{align} 
    \comm{A}{B} = \sum_\ell a_\ell \hat{A}_\ell \hat{B}_\ell, 
\end{align}
and $\sum_\ell \abs{a_\ell} \leq 4 e^{-d_{ij}/4\kappa}$,
where $d_{ij}$ is the distance between $i$ and $j$.
\end{lemma}
\begin{proof}
Let $A = \sum_{r} A_{i,r}$ and $B = \sum_{r} B_{i,r}$ be the decompositions of $A$ and $B$ according to \cref{def:concentrated-operator}.
We divide $A$ into a sum of
$A_{<} = \sum_{r < d_{ij}/2} A_{i,r}$ and $A_{\geq} = \sum_{r \geq d_{ij}/2} A_{i,r}$ based on the radius $r$ and divide $B = B_{<} + B_{\geq}$ similarly. 
The commutator can then be expanded as
\begin{align} 
    \comm{A}{B} = \comm{A_<}{B_\geq} + \comm{A_\geq}{B}. 
\end{align}
Note that $\comm{A_<}{B_<}$ vanishes because the operators are supported on distinct regions of the lattice.

Note that $A_<, A_\geq$ are both in $\mathcal C_i(c,\kappa) \subseteq C_i(c,2\kappa)$.
We can, however, prove a stronger statement for $A_\geq$.
Since $r \geq d_{ij}/2$, the identity $e^{-r/\kappa}\leq e^{-d_{ij}/4\kappa} e^{-r/2\kappa}$ implies $\tilde A_\geq \equiv e^{d_{ij}/4\kappa} A_\geq \in C_i(c,2\kappa)$.
Repeating the same argument for $B_\geq$, we have
\begin{align} 
     \comm{A}{B} = e^{-d_{ij}/4\kappa} \left(\comm{A_<}{\tilde B_\geq} 
     + \comm{\tilde A_\geq}{ B}
     \right).
\end{align}
Expanding the commutators, we arrive at \cref{lem:comm-decomposition} with $a_\ell = \pm e^{-d_{ij}/4\kappa}$,  and $\hat{A}_1 \equiv A_<$,  $ \hat{A}_2 \equiv \tilde A_{\geq} $ ,  $\hat{B}_1 \equiv B_{\geq}$, and  $\hat{B}_2 \equiv B$.
\end{proof}

\begin{definition}[Concentrated $k$-local Strings]
An operator $A$ is a concentrated $k$-local string if there exist constants $c, \kappa$ such that
\begin{align} 
    A = A^{(1)}_{i_1} A^{(2)}_{i_2} \dots A^{(k)}_{i_k},         
\end{align}    
where $A^{(j)}_{i_j} \in \mathcal C_{i_j}(c,\kappa)$ for all $j = 1,\dots,k$.
We use $S(A) = \{i_1,\dots,i_k\}$ to denote the set of centers.
\end{definition}

\begin{definition}[Concentrated $k$-local Hamiltonians]\label{def:concentrated-k-local-H}
A Hamiltonian $H = \sum_{\mu} a_\mu {H_\mu}$ is a concentrated $k$-local Hamiltonian with an effective interaction degree $d$ if
\begin{enumerate}
    \item $H_\mu = A^{(\mu,1)}_{i_{\mu,1}} A^{(\mu,2)}_{i_{\mu,2}} \dots A^{(\mu,k)}_{i_{\mu,k}}$ are concentrated $k$-local strings for all $\mu$, where $A^{(\mu,j)}_{i_{\mu,j}}\in \mathcal C_{i_{\mu,j}}(c,\kappa)$, and
    \item for every site $l$ on the lattice,
    \begin{align} 
         \sum_{\mu: l\in S(H_\mu)} \abs{a_\mu} \leq d. 
    \end{align}
\end{enumerate}
We use $\mathcal H(c,\kappa,k,d)$ to denote the set of such Hamiltonians. Note that the definition of $H$ is invariant under the following rescaling: $c \to g c$ and $d \to d/g^k$ for some constant $g$.  
\end{definition}

Note that the definition of the effective interaction degree in \cref{def:concentrated-k-local-H} is slightly different from the one defined in \cref{sec:setup}. 

We also note that a geometrically local Hamiltonian is also a concentrated $1$-local Hamiltonian:
\begin{lemma}\label{lem:geometrically-local-is-also-concentrated}
Let $B$ be a geometrically local with an interaction range $\chi$. 
It follows that $B\in \mathcal H(c,\kappa,1,1)$ for $c = e^{\chi/\kappa}$ and $\kappa$ is any constant. 
\end{lemma}

Concentrated $k$-local Hamiltonians are generalization of $k$-local Hamiltonians.
The following lemma upper bounds the norm of a concentrated $k$-local Hamiltonian.
\begin{lemma}\label{lem:norm-of-concentrated-H}
If $H \in \mathcal H(c,\kappa,k,d)$ on $n$ qubits, then $\norm{H} \leq c^k (\kappa+1)^k n d$.    
\end{lemma}
\begin{proof}
We have
\begin{align} 
    \norm{H} \leq \sum_\mu \abs{a_\mu}\norm{H_\mu}
    \leq c^k(\kappa+1)^k \sum_\mu \abs{a_\mu} 
    \leq c^k(\kappa+1)^k \sum_i \sum_{\mu: i \in S(H_\mu)} \abs{a_\mu}
    \leq c^k(\kappa+1)^k n d,
\end{align}
where in the second inequality we used that $\Vert A_{i_{\mu, l}}^{(\mu, l)}\Vert \leq c (1+\kappa), \forall A_{i_{\mu, l}}^{(\mu, l)} \in C_{i_{\mu,j}}(c,\kappa)$, and in the final inequality we overcounted the sum over $\mu$ by summing over index $i$ and $\mu: i \in S(H_{\mu})$. 
\end{proof}

\begin{lemma}\label{lem:commutator-between-H-K}
If $H\in \mathcal H(c,\kappa,k_1,d_1)$ and $K \in \mathcal H(c,\kappa,k_2,d_2)$, then $\comm{H}{K} \in \mathcal H(c',\kappa',k',d')$ for 
\begin{align}
c' &= c (8(4\kappa+1) k_1 k_2 )^{1/(k_1+k_2)} ~,\\
\kappa' &= 2\kappa~,\\ 
k' &= k_1+k_2~, \\
d' &=d_1d_2~.
\end{align} 
\end{lemma}
\begin{proof}

Let $H = \sum_\mu a_\mu H_\mu$ and $K = \sum_\nu b_\nu K_\nu$ be decompositions into concentrated $k_1$-local and $k_2$-local strings.
Let $V_{\mu,\nu} = a_\mu b_\nu \comm{H_\mu}{K_\nu} = \comm{A^{(\mu,1)}_{i_{\mu, 1}} A^{(\mu,2)}_{i_{\mu, 2}} \dots A^{(\mu,k_1)}_{i_{\mu, k_1}}}{B^{(\nu,1)}_{j_{\nu, 1}} B^{(\nu,2)}_{j_{\nu, 2}} \dots B^{(\nu,k_2)}_{j_{\nu, k_2}}}$.
Using the chain rule for commutators, we have
\begin{align} 
    V_{\mu,\nu}
    = a_\mu b_\nu\sum_{u,v}  
    A^{(\mu,1)}_{i_{\mu, 1}}  \dots A^{(\mu,u-1)}_{i_{\mu, u-1}}
    B^{(\nu,1)}_{j_{\nu, 1}}  \dots B^{(\nu,v-1)}_{j_{\nu, v-1}}  
    \comm{ A^{(\mu,u)}_{i_{\mu, u}}}{B^{(\nu,v)}_{j_{\nu, v}}}
    B^{(\nu,v+1)}_{j_{\nu, v+1}}  \dots B^{(\nu,k_2)}_{j_{\nu, k_2}}
    A^{(\mu,u+1)}_{i_{\mu,u+1}}  \dots A^{(\mu,k_1)}_{i_{\mu,k_1}}.\label{eq:long-comm-chain-rule}
\end{align}
Since $A^{(\mu,u)}_{i_{\mu, u}}$ and $ B^{(\nu,v)}_{j_{\nu, v}}$ are concentrated operators, \cref{lem:comm-decomposition} implies a decomposition
\begin{align} 
    \comm{ A^{(\mu,u)}_{i_u}}{B^{(\nu,v)}_{j_{v}}} = 
    \sum_{\ell_{\mu,\nu,u,v}}
    \beta_{\ell_{\mu,\nu,u,v}}
    \hat{ A}^{(\ell_{\mu,\nu,u,v})}_{i_\mu} \hat{ B}^{(\ell_{\mu,\nu,u,v})}_{j_\nu}
\end{align}
such that $\hat {A}^{(\ell_{\mu,\nu,u,v})}_{i_\mu} \in \mathcal C_{i_\mu}(c,2\kappa)$, $\hat {B}^{(\ell_{\mu,\nu,u,v})}_{j_\nu} \in \mathcal C_{j_\nu}(c,2\kappa)$, and 
$\sum_{\ell_{\mu,\nu,u,v}} \abs{\beta_{\ell_{\mu,\nu,u,v}}} = 4 e^{-d_{i_\mu j_\nu}/4\kappa}$.
Therefore, the commutator $\comm{H}{K}$ can be decomposed as
\begin{align} 
    \comm{H}{K}
    = \sum_{\mu\nu} a_\mu b_\nu \sum_{u,v} \sum_{\ell_{\mu,\nu,u,v}} \beta_{\ell_{\mu,\nu,u,v}} V_{\mu,\nu,u,v,\ell_{\mu,\nu,u,v}}
    = \sum_\tau c_\tau V_\tau \label{eq:commutator-lemma5}
\end{align}
for some concentrated $(k_1+k_2)$-local strings $V_\tau = V_{\mu,\nu,u,v,\ell_{\mu,\nu,u,v}}$.
Here, we use $\tau$ to denote a tuple of $\mu,\nu,u,v,\ell_{\mu,\nu,u,v}$.

Next, we prove the second condition of \cref{def:concentrated-k-local-H} for some $d'$ to be determined.
A fixed site $i$ is in $S(V_{\mu,\nu,u,v,\ell_{\mu,\nu,u,v}})$ only if it is in either $S(H_\mu)$ or $S(K_\nu)$.
Suppose $i \in S(H_\mu) = \{i_1,\dots,i_{k_1}\}$.
For each $\mu, u, v$, we have
\begin{align} 
     \sum_{\nu} \sum_{\ell_{\mu,\nu,u,v}} \abs{a_\mu}\abs{b_\nu} \abs{\beta_{\ell_{\mu,\nu,u,v}}}
     = 4 \sum_{\nu}  \abs{a_\mu}\abs{b_\nu}  e^{-d_{i_\mu j_\nu}/4\kappa}
     \leq  4 \abs{a_\mu} \sum_{j\in \Lambda} e^{-d_{i_\mu j}/4\kappa} \sum_{\nu:j_\nu = j}  \abs{b_\nu}
     \leq 4(4\kappa+1) \abs{a_\mu} d_2.   
\end{align}

We then upper bound the sums over $u, v$ by $k_1, k_2$ respectively, and noting that the sum over $\mu$ is restricted to $\mu$ satisfying $i \in S(H_\mu)$, we have
\begin{align} 
    \sum_{\tau: i \in S(V_{\tau})}   \abs{c_\tau} \leq \left(\sum_{\tau: i \in S(H_{\mu})}   \abs{c_\tau}+ \sum_{\tau: i \in S(K_{\mu})}   \abs{c_\tau}\right)  \leq 8(4\kappa+1) k_1 k_2 d_1 d_2,\label{eq:bound-lemma5}
\end{align}
where a factor of two comes from the case when $i \in S(K_{\nu})$, and a factor of $d$ comes from the fact that $\sum_{\mu} |a_{\mu}| \leq d_1$. 

Therefore, by combining \cref{eq:commutator-lemma5} and \cref{eq:bound-lemma5} and using the invariance of $\mathcal H$ under rescaling, we get that 
\begin{align}
    [H,K] \in \mathcal H(c,2\kappa,k_1+k_2,8(4\kappa+1) k_1 k_2 d_1 d_2)
    = \mathcal H(c (8(4\kappa+1) k_1 k_2)^{1/(k_1+k_2)},2\kappa,k_1+k_2, d_1 d_2)~.
\end{align} 
\end{proof}

The following lemma is also useful for counting the support of a nested commutator between concentrated 1-local Hamiltonians.
\begin{lemma}\label{lem:commutators-between-1-local-Ham}
Let $H^{(1)}, \dots, H^{(q)} \in \mathcal H(c,\kappa,1,1)$ be $q$ concentrated $1$-local Hamiltonians.
For any integers $R \geq 1$, there exists $V = \sum_{i\in\Lambda} V_i$ such that
\begin{itemize}
        \item $V \in \mathcal H(c_q,2^q\kappa,q,1)$, where $c_q$ is a constant depending on only $q, \kappa$,
        \item $V_i$ is supported on at most $q + R$ qubits, and
        \item $\norm{V - \comm{H^{(q)}}{\dots\comm{H^{(2)}}{H^{(1)}}}} = \O{n e^{-R/4\kappa}}$.
    \end{itemize}    
\end{lemma}
\begin{proof}
Let $H^{(k)} = \sum_{i\in\Lambda} H^{(k)}_i$ be the decomposition of $H^{(k)}$ into concentrated operators,
where $H^{(k)}_i \in \mathcal C_i(c,\kappa)$.
Define
\begin{align} 
    W_i = \sum_{i}\comm{H^{(q)}}{\dots\comm{H^{(2)}}{H^{(1)}_i}}
    = \sum_{i}\sum_{i_2,\dots,i_q \in\Lambda}\sum_{r_1,\dots,r_q = 0}^{\infty}
    \comm{H^{(q)}_{i_q,r_q}}{\dots\comm{H^{(2)}_{i_2,r_2}}{H^{(1)}_{i,r_1}}},
\end{align}
where we have further decomposed concentrated operators using their definitions.
Let $g(i,i_2,\dots,i_q,r_1,\dots,r_q)$ be the number of qubits in the joint support of $H^{(q)}_{i_q,r_q},\dots, H^{(2)}_{i_2,r_2},H^{(1)}_{i,r_1}$.
From $W_i$, we construct $V_i$ by discarding summands where $g(i,i_2,\dots,i_q,r_1,\dots,r_q)>q+R$.
\Cref{lem:commutator-between-H-K} implies $V = \sum_i V_i \in \mathcal H(c_q,2^q \kappa,q,1)$ for a constant $c_q$ that may depend on $q$.
By construction, $V_i$ is supported on at most $q+R$ qubits.

To bound the error $\norm{V_i - W_i}$, we recall $H^{(1)}_{i,r_1}$ is supported on at most $2r_1$ qubits around $i$.
Similarly,  $H^{(2)}_{i_2,r_2}$ is supported on at most $2r_2$ qubits around $i_2$.
Let $m + q = g(i,i_2,\dots,i_q,r_1,\dots,r_q)$. 
Besides the $q$ centers, the other $m$ qubits come from the balls of radii $r_1,\dots,r_q$. It follows that
\begin{align} 
    m \leq 2(r_1+r_2+\dots+r_q). 
\end{align}
Therefore,
\begin{align} 
    \norm{\comm{H^{(q)}_{i_q,r_q}}{\dots\comm{H^{(2)}_{i_2,r_2}}{H^{(1)}_{i,r_1}}}}
    \leq 2^q c^q e^{-r_1/\kappa}\dots e^{-r_q/\kappa} \leq 2^q c^q e^{-m/2\kappa}.  \label{eq:for-each-m}
\end{align}
So each summand that we dropped will be bounded by $2^q c^q e^{-m/2\kappa}$.

Next, we count the number of such summand, i.e. the number of the choices of $r_1,\dots,r_q$ and $i_2,\dots,i_q$ such that $g(i,i_2,\dots,i_q,r_1,\dots,r_q) = m+q$.
Denote this number by $f(m)$.
Note that $m \geq 2 r_1$, so $0\leq r_1 \leq m/2$ and the number of choices for $r_1$ is at most $m/2 + 1\leq m$, assuming $m \geq 2$.
Using a similar argument for $r_2,\dots,r_q$, the number of choices of these radii is at most $m^q$.

To count the number of centers $i_2,\dots,i_q$, we note that $d_{i,i_2}\leq r_1 + r_2 \leq m$ (otherwise, the commutator would vanish).
Similarly, $d_{i,i_3}\leq r_1 + 2r_2 + r_3 \leq 2m$.
Therefore, the number of choices for $i_2,\dots,i_q$ is at most
$(2m)(4m)(6m)\dots(2(q-1)m) \leq 2^q q! m^q$.
Combining with the counting argument for $r_1,\dots,r_q$, we have
\begin{align} 
    f(m) \leq 2^q q! m^{2q}. \label{eq:number-of-m} 
\end{align}

Combining \cref{eq:number-of-m,eq:for-each-m}, the error of approximating $W_i$ with $V_i$ is
\begin{align} 
    \norm{V_i - W_i}
    \leq \sum_{m = R+1}^{\infty} 4^q q! c^q m^{2q} e^{-m/2\kappa} 
    \leq \sum_{m = R+1}^{\infty} 4^q q!(2q!)(4\kappa)^{2q} c^q e^{-m/4\kappa}
    = \O{e^{-R/4\kappa}},
\end{align}
where we have used $m^{2q}\leq (2q)!(4\kappa)^{2q} e^{m/4\kappa}$ for $m\geq 2$.
Since $\sum_{i\in\Lambda} W_i = \comm{H^{(q)}}{\dots\comm{H^{(2)}}{H^{(1)}}}$, we arrive at the lemma.
\end{proof}

\section{The Lieb-Robinson bound and the locality of Magnus operators}\label{app:magnus-operator-structure}
In this section, we use the Lieb-Robinson bound and the machinery developed in \cref{app:concentrated-operator}
to bound the error of the Magnus operator.
We begin by showing that a $k$-local operator evolved under a geometrically local Hamiltonian can be written as a concentrated operator.

\begin{lemma}
\label{lem:decomposing-using-LR}
Let $B$ be a unit-norm operator supported on a single site $i$ and $A$ is a geometrically local Hamiltonian.
There exists constant $c, \kappa$ such that
\begin{align} 
    e^{iAt} B e^{-iAt} \in \mathcal C_i(c,\kappa). 
\end{align}
\end{lemma}
\begin{proof}
Given an operator $O$, we define a superoperator
\begin{align} 
    \mathbb P_{i,r} O \equiv \int_{U\sim \mathcal U_r} dU U^\dag O U, \label{eq:projector}
\end{align}
where the integral is over the Haar measure of unitaries supported entirely outside a ball of radius $r$ from~$i$.
Consequently, $\mathbb P_{i,r} O$ is supported entirely within this ball.

Let $B(t) = e^{iAt} B e^{-iAt}$.
The Lieb-Robinson bound implies that $\mathbb P_{i,r} B(t)$ is a good approximation to $B(t)$ in the limit of large $r$~\cref{eq:lb-bound}.
Indeed, we have
\begin{align} 
    \norm{B(t) - \mathbb P_{i,r} B(t)}
    \leq \int_U dU~ \norm{\comm{U}{B(t)}}
    \leq \int_U dU~ K e^{vt - r/\chi}
    =  K e^{vt - r/\chi},
\end{align}
for some constant $K, v, \chi$.

We define
\begin{align} 
    B_{i,r} = \mathbb P_{i,r} B(t) - \mathbb P_{i,r-1} B(t) 
\end{align}
for $r = 0,1,\dots$.
Here, we define $\mathbb P_{i,-1} B(t) = 0$ for convenience.
Then, we have
\begin{align} 
    \norm{B_{i,r}} \leq \norm{B(t) - \mathbb P_{i,r} B(t)}
    + \norm{B(t) - \mathbb P_{i,r-1} B(t)} 
    \leq 2Ke^{vt - 1/\chi} e^{-r/\chi}.
\end{align}
Therefore, $B(t) \in \mathcal C_i(c,\kappa)$ with $c = 2K e^{vt-1/\chi}$ and $\kappa = \chi$.
\end{proof}

\begin{lemma}\label{lem:evolved-concentrated-1-local}
Let $A$ be a geometrically local Hamiltonian with an interaction range $\chi\geq 1$ and $B \in \mathcal H(c,\kappa,1,1)$ be a concentrated $1$-local Hamiltonian with $\kappa \leq \chi$.
We have
\begin{align} 
          e^{iAt} B e^{-iAt} \in \mathcal H(c',2\chi,1,1),
\end{align}   
for $c' = 16Kc \chi^2 e^{vt-1/\chi}$ for a constant $K$.
In particular, $c'$ is a constant if $t = \mathcal O(1)$.
\end{lemma}
\begin{proof}
Let $B = \sum_{i\in \Lambda} B_i = \sum_{i\in\Lambda} \sum_{r = 0}^{\infty} B_{i,r}$ be a decomposition of $B$ into concentrated operators, where $B_i\in\mathcal C_i(c,\kappa)$.
For each $i$, let $B_{i,r}(t) = e^{iAt} B_{i,r} e^{-iAt}$.
We have
\begin{align} 
     e^{iAt} B_i e^{-iAt} = \sum_{r = 0}^{\infty} 
     \left[\mathbb P_{i,r} B_{i,r}(t) + \sum_{s = 1}^{\infty} (\mathbb P_{i,r+s} - \mathbb P_{i,r+s-1}) B_{i,r}(t) \right],
\end{align}
where $\mathbb P_{i,s}$ is defined in \cref{eq:projector}.
Let $V_{i,r,s} =(\mathbb P_{i,r+s} - \mathbb P_{i,r+s-1}) B_{i,r}(t) $ for $s \geq 1$ and $V_{i,r,0} = \mathbb P_{i,r} B_{i,r}(t)$.
It follows from the Lieb-Robinson bound that
\begin{align} 
    \norm{V_{i,r,s}} \leq 2K e^{vt-1/\chi} (2r+1) e^{-s/\chi}\norm{B_{i,r}}
    \leq 2Kc e^{vt-1/\chi} (2r+1) e^{-(r+s)/\chi} ,
\end{align}
where the factor $2r+1$ comes from the size of the support of $B_{i,r}(0)$.
Under this notation,
\begin{align} 
   B_i^{'}\equiv  e^{iAt} B_i e^{-iAt} = \sum_{r = 0}^{\infty} \sum_{s = 0}^{\infty} V_{i,r,s},
\end{align}
where $V_{i,r,s}$ is supported entirely within a ball of radius $r+s$ around $i$.
Next, we group $V_{i,r,s}$ based on $t = r+s$.
Define
\begin{align} 
    B'_{i,t} = \sum_{r = 0}^{t} V_{i,r,t-r}.   
\end{align}
It follows that $B'_{i,t}$ is supported within a ball of radius $t$ from $i$ and
\begin{align} 
     \norm{B'_{i,t}}
     \leq 2Kc e^{vt-1/\chi} \sum_{r = 0}^{t} (2r+1) e^{-t/\chi}
     =16Kc \chi^2 e^{vt-1/\chi}  e^{-t/2\chi}
\end{align}
where we have used $\sum_{r = 0}^{t} (2r+1) \leq 8\chi^2 e^{-t/2\chi}$ for all $\chi \geq 1$.
Therefore $B'_{i} \in\mathcal C_i(c',2\chi)$ with $c' = 16Kc \chi^2 e^{vt-1/\chi}$ and the lemma follows.
\end{proof}

Several corollaries follow immediately:
\begin{corollary}\label{cor:BI}
If $B$ is a $k$-local Hamiltonian with an effective interaction degree $d$ and $A$ is a geometrically local Hamiltonian with an interaction range $\chi$, then 
\begin{align} 
        e^{iAt} B e^{-iAt} \in \mathcal H(c,\kappa,k,d), 
\end{align}    
for some $c$ that may depend on $t$ and a constant $\kappa$.
Moreover, if $t = O(1)$, then $c$ is a constant. 
\end{corollary}

\cref{cor:BI,lem:commutator-between-H-K} imply
\begin{corollary}\label{cor:Omega-structure}
Let $ B$ be a $k$-local Hamiltonian with an effective interaction degree $d$ and $A$ be a geometrically local Hamiltonian.
Let $\Omega_q(t)$ be the $q$th-order Magnus operator generating the interaction-picture evolutions of $H = A + \alpha B$.
For $t = O(1)$, we have
\begin{align} 
        &\Omega_q(t) \in \mathcal H(c_q,\kappa_q,q k,\alpha^q d^q t^q), \text{ and} \label{eq:bound-on-Omega-q}\\
        &\Omega_q^{\prime}(t) \in \mathcal H(c_q,\kappa_q,q k,\alpha^q d^q t^{q-1})
\end{align}    
for some constants $c_q$ and $\kappa_q$.
In addition, \cref{eq:bound-on-Omega-q} implies
\begin{align} 
     \sum_{j = 1}^q \Omega_j(t) \in \mathcal H(c_q^{\prime},\kappa_q^{\prime},q k,\alpha td)
\end{align}
for $\alpha d t \leq 1$.
\end{corollary}

\Cref{lem:commutators-between-1-local-Ham,cor:BI,lem:evolved-concentrated-1-local} imply a more robust structure of the Magnus operators:
\begin{corollary}\label{cor:Omega-structure-geometrically-local}
Let $ B$ and $A$ be geometrically local Hamiltonians and $\chi$ be the interaction range of $A$.
Let $\Omega_q(t)$ be the $q$th-order Magnus operator generating the interaction-picture evolutions of $H = A + \alpha B$.
For $t = O(1)$ and for any $R \geq 1$, there exists an operator $\Omega_{\text{loc},q} = \sum_{i\in\Lambda}\Omega_{\text{loc},q,i}$ such that
\begin{itemize}
    \item $\Omega_{\text{loc},q}\in \mathcal H(c_q,2^{q+1} \chi,q,1)$ for a constant $c_q$ depending on only $q,\chi$,
    \item $\Omega_{\text{loc},q,i}$ is supported on at most $q + R$, and
    \item $\norm{\Omega_{q}-\Omega_{\text{loc},q}} = \O{n t^q \alpha^q e^{-R/8\chi}}$.  
\end{itemize}
\end{corollary}

\bibliographystyle{unsrt}
\bibliography{Ref,zotero}

\pagebreak

\newpage

		\end{document}